\documentclass[preprint,12pt]{elsarticle}
\usepackage{amsthm,amsmath,amssymb}
\usepackage{booktabs}
\usepackage{tabularx}
\usepackage{multirow}
\usepackage{times}
\usepackage{hyperref}
\usepackage[utf8]{inputenc}
\usepackage{natbib}

\usepackage{colortbl}
\usepackage{xcolor}

\usepackage{geometry}
\journal{}

\newtheorem{observation}{Observation}
\newtheorem{corollary}{Corollary}

\newtheorem{lemma}{Lemma}
\newtheorem{theorem}{Theorem}

\newtheorem{rerule}{Reduction Rule}

\newcommand{\np}{\textsf{NP}}
\newcommand{\nph}{{\np}-hard}
\newcommand{\nphns}{{\np}-hardness}
\newcommand{\fpt}{\textsf{FPT}}
\newcommand{\wa}{\textsf{W[1]}}
\newcommand{\wb}{\textsf{W[2]}}
\newcommand{\wbh}{\wb-hard}

\newcommand{\poly}{\textsf{P}}
\newcommand{\wbhns}{\wb-hardness}
\newcommand{\xp}{\textsf{XP}}

\newcommand{\yes}{YES}

\newcommand{\yesins}{{\yes}-instance}
\newcommand{\setmid}{\mid}
\newcommand{\abs}[1]{|#1|}
\newcommand{\bigo}[1]{O(#1)}
\newcommand{\smallo}[1]{o(#1)}
\newcommand{\prob}[1]{{\sc{#1}}}
\newcommand{\memph}[1]{\text{#1}}
\newcommand{\bigos}[1]{O^*(#1)}

\newcommand{\arc}[2]{(#1, #2)}
\newcommand{\ARC}[1]{E(#1)}
\newcommand{\lsr}{\text{LSR}}
\newcommand{\csr}{\text{CSR}}

\renewenvironment{proof}[1][{\it{Proof}}]{\noindent\textbf{#1} }{\hfill \qed\medskip}

\newcommand{\EP}[3]
{
\begin{center}
{
\small
\begin{tabular}{|ll|} \hline
\multicolumn{2}{|l|}{\textsc{#1}} \\[2mm]
{\bf Input:}    & \parbox[t]{0.85\columnwidth}{#2\vspace*{1mm}}  \\
{\bf Question:} & \parbox[t]{0.85\columnwidth}{#3\vspace*{1mm}} \\ \hline
\end{tabular}
}
\end{center}
}

\begin{document}
	
\begin{frontmatter}
		
		\title{Group Control for Procedural Rules: Parameterized Complexity and Consecutive Domains}
		
		\author{Yongjie YANG\corref{cor1}}\ead{yyongjiecs@gmail.com}
        \author{Dinko DIMITROV}\ead{dinko.dimitrov@mx.uni-saarland.de}

\address{Chair of Economic Theory, Saarland University, Saarbr\"{u}cken, 66123, Germany}

\begin{abstract}
We consider {\prob{Group Control by Adding Individuals}} (\prob{GCAI}) in the setting of group identification for two procedural rules---the consensus-start-respecting rule and the liberal-start-respecting rule.
It is known that {\prob{GCAI}} for both rules are {\nph}, but whether they are fixed-parameter tractable with respect to the number of distinguished individuals remained open. We resolve both open problems in the affirmative. In addition, we strengthen the {\nphns} of {\prob{GCAI}} by showing that, with respect to the natural parameter the number of added individuals, {\prob{GCAI}} for both rules are {\wbh}. Notably, the {\wbhns} for the liberal-start-respecting rule holds even when restricted to a very special case where the qualifications of individuals satisfy the so-called consecutive ones property. However, for the consensus-start-respecting rule, the problem becomes polynomial-time solvable in this special case. We also study a dual restriction where the disqualifications of individuals fulfill the consecutive ones property, and show that under this restriction {\prob{GCAI}} for both rules turn out to be polynomial-time solvable. Our reductions for showing {\wbhns} also imply several lower bounds concerning kernelization and exact algorithms.
\end{abstract}
		
		
\end{frontmatter}

\section{Introduction}
In the model of group identification, we have a group of individuals each of whom holds binary valuations on all individuals including herself, and the model aims to determine who among these individuals are socially qualified by utilizing a certain social aggregation rule (cf.~\cite{KasherRwhoisj}). Since the initial works of Kasher~\cite{Kasher1993} and Kasher and Rubinstein~\cite{KasherRwhoisj}, group identification has been extensively explored from the perspective of economics, with the main focus being on axiomatic characterizations of different social aggregation rules (cf.\ \cite{DBLP:series/sfsc/Dimitrov11,DBLP:journals/mss/DimitrovSX07,DBLP:journals/mss/Nicolas07,DBLP:journals/geb/Miller08,DBLP:journals/jet/SametS03}).

As social aggregation rules can be seen as special voting systems where the individuals are both voters and candidates, inspired by the pioneering work of Bartholdi, Tovey, and Trick~\cite{Bartholdi92howhard} on voting control problems, Yang and Dimitrov~\cite{DBLP:journals/aamas/YangD18} initiated the study of group identification from a computer science perspective by investigating the complexity of the {\prob{Group Control By Adding/Deleting Individuals}} ({\prob{GCAI}}/{\prob{GCDI}}) problems. In particular, {\prob{GCAI}}/{\prob{GCDI}}  consists in determining if a given (valuation) profile can be modified by adding/deleting a limited number of individuals to make a given subset of distinguished individuals all socially qualified.
Yang and Dimitrov studied the
consensus-start-respecting rule and the liberal-start-respecting rule, and showed that  {\prob{GCAI}} for both rules are {\nph}, while {\prob{GCDI}} for both rules turned out to be polynomial-time solvable. However, it is left open whether {\prob{GCAI}} for these two rules are fixed-parameter tractable (\fpt) with respect to the number of distinguished candidates. We resolve the open questions in the affirmative by reducing {\prob{GCAI}} to a variant of the {\prob{Directed Steiner Tree}} problem (Theorems~\ref{thm-csr-fpt} and~\ref{thm-gcai-lsr-fpt}).
In addition, we strengthen the above-mentioned {\nphns} results by showing that {\prob{GCAI}} for both rules are {\wbh} with respect to the number of added individuals (Theorems~\ref{thm-gcai-csr-wbh} and \ref{thm-gcai-lsr}). Particularly, the {\wbhns} for the liberal-start-respecting rule holds even when restricted to a very special case where the qualifications of individuals satisfy the so-called consecutive ones property. However, for the consensus-start-respecting rule, the problem is polynomial-time solvable in this special case (Theorem~\ref{poly-gcai-csr-qc}). We also study a dual restriction where the disqualifications fulfill the consecutive ones property, and show that under this restriction {\prob{GCAI}} for both rules turn out to be polynomial-time solvable (Theorem~\ref{thm-gcai-lsr-dqc}). Our hardness reductions also lead to numerous lower bounds concerning kernelizations and exact algorithms (Corollaries~\ref{cor-1}--\ref{cor-last}).

\section{Related Works}

Since the first work of Yang and Dimitrov~\cite{DBLP:journals/aamas/YangD18} on the complexity of group control problems, several related problems have been proposed and studied very recently. It should be pointed out that in addition to the consensus-start-respecting rule and the liberal-start-respecting rule, Yang and Dimitrov~\cite{DBLP:journals/aamas/YangD18} also studied the class of consent rules proposed by Samet and Schmeidler~\cite{DBLP:journals/jet/SametS03}, and established a complete complexity landscape of the group control problems with respect to the consent quotes of these rules.
Erd\'{e}lyi, Reger, and Yang~\cite{DBLP:journals/aamas/ErdelyiRY20,DBLP:conf/aldt/ErdelyiRY17} studied the destructive counterpart of the group control problems, group bribery problems, and the problems of determining socially qualified individuals with incomplete information. Later, Erd\'{e}lyi and Yang~\cite{DBLP:conf/atal/Erdelyi020} studied the complexity of microbribery in group identification. Additionally, Boehmer~et~al.~\cite{DBLP:conf/ijcai/BoehmerBKL20} also considered numerous bribery problems in group identification through the lens of parameterized complexity. Junker~\cite{Junker2022,Junker2022b} later extended the results of Boehmer~et~al.~and studied some other variants of strategic problems in group identification. Motivated by applications in information diffusion in social networks, Bla\v{z}ej~\cite{Blaej2022} studied some generalization of group control problems for the two procedural rules mentioned above. Very recently, Yang and Dimitrov~\cite{YangDimitrov2022MSS} studied group control problems for the class of consent rules when restricted to cyclic domains which contain the domains studied in the paper as special cases. They showed that these problems, being computationally hard to solve in general, become polynomial-time solvable when the input profile falls into the category of cyclic domains.

Voting problems restricted to special domains have been widely studied in the literature. Particularly, the consecutive domain has been studied as an analog of single-peaked domain for dichotomous preferences (cf.~\cite{DBLP:journals/jair/BrandtBHH15,DBLP:journals/iandc/FaliszewskiHHR11}). This domain has been also studied under the name {\textit{candidate interval}}  (cf.~\cite{DBLP:conf/aaai/Peters18,DBLP:conf/ijcai/ElkindL15,DBLP:conf/atal/LiuG16,DBLP:conf/ijcai/Yang19a}). It is known that many voting problems which are {\nph} in general become polynomial-time solvable when restricted to this domain, with only a few exceptions (cf.~\cite{DBLP:journals/jair/BetzlerSU13,DBLP:journals/tcs/SkowronYFE15}). For further discussions on the complexity of voting problems in restricted domains, we refer to~\cite{structuredpreferencesElkindLP,Hemaspaandra2016,DBLP:journals/aarc/Karpov22,DBLP:journals/corr/abs-2205-09092} for comprehensive surveys.

It should be also noted that the consecutive domain is equivalent to the so-called consecutive ones property of $(0,1)$-matrices which has a long research history and has found significant applications in a broad range of areas (cf.\ the survey~\cite{DBLP:journals/eatcs/Dom09} and references therein). Recall that a $(0,1)$-matrix satisfies the consecutive ones property if its columns can be permuted so that in each row all~$1$s are consecutive.

\section{Preliminaries}
\label{sec_notation}
Throughout this paper we will need the following basic ingredients.
For an integer~$i$,~$[i]$ is the set of positive integers no greater than~$i$.

\subsection{Social Aggregation Rules}
Let~$N$ be a set of~$n$ individuals. Each individual $a\in N$
has an opinion who from the set~$N$ possess a certain qualification and who
do not. For $a^{\prime }\in N$, we write $\varphi (a,a^{\prime})=1$ if~$a$ qualifies~$a^{\prime}$, and write $\varphi (a,a^{\prime})=0$ if~$a$ disqualifies~$a^{\prime}$. The mapping $\varphi : N\times N\rightarrow \left\{ 0,1\right\} $ is called a \textit{profile}
over~$N$. A \textit{social aggregation rule} is a function~$f$ assigning a subset $f(\varphi ,T)\subseteq T$ to each pair $(\varphi ,T)$ of a
profile~$\varphi $ over~$N$ and a subset $T\subseteq N$. The members of $f(\varphi ,T)$ are called~$f$ {\it{socially qualified individuals}} in~$T$ at the profile~$\varphi$.

In what follows we focus in our analysis on two procedural rules: the consensus-start-respecting
rule and the liberal-start-respecting rule. The reader is referred to~\cite{DBLP:journals/mss/DimitrovSX07} for
axiomatic characterizations of these rules.

\begin{description}
\item[Consensus-start-respecting rule~$f^{\csr}$]
This rule determines the socially qualified individuals iteratively. First, all individuals qualified by everyone are considered socially qualified. Then, in each
  iteration, all individuals who are qualified by at least one of
  the currently socially qualified individuals are added to the set of
  socially qualified individuals. The iterations terminate when no new
  individuals can be added this way. Formally, for every~$T\subseteq N$, let
  \begin{displaymath}
    K_0^{\text{C}}(\varphi, T)=\{a \in T \mid \forall(a' \in T)[\varphi(a', a) =1]\}.
  \end{displaymath}
  For each~$\ell=1$,~$2$, \dots, let $K_\ell^{\text{C}}(\varphi, T)=K_{\ell-1}^{\text{C}}(\varphi, T) \cup \{a \in T \mid \exists (a' \in K_{\ell-1}^{\text{C}}(\varphi, T))[\varphi(a', a) =1]\}$. 
  Then, $f^{\csr}(\varphi, T)=K_\ell^{\text{C}}(\varphi, T)$ for some integer~$\ell$ such
  that $K_\ell^{\text{C}}(\varphi,T)=K_{\ell-1}^{\text{C}}(\varphi, T)$.

\item[Liberal-start-respecting rule~$f^{\lsr}$] This rule
  is analogous to $f^{\csr}$ with only the difference that the initial
  socially qualified individuals are those who qualify themselves. In
  particular, for every~$T\subseteq N$, let $K_0^{\text{L}}(\varphi, T)=\{a \in T\mid \varphi(a, a)=1\}$. 
  For each integer~$\ell=1$,~$2$, \dots, let $K_\ell^{\text{L}}(\varphi, T)=K_{\ell-1}^{\text{L}}(\varphi, T) \cup \{a \in T \mid \exists (a' \in K_{\ell-1}^{\text{L}}(\varphi, T))[\varphi(a', a) =1]\}$. 
  Then, $f^{\lsr}(\varphi, T)=K_\ell^{\text{L}}(\varphi, T)$ for some integer~$\ell$ such
  that $K_\ell^{\text{L}}(\varphi, T)=K_{\ell-1}^{\text{L}}(\varphi, T)$.
\end{description}
It should be noted that when $K_0^{\text{C}}(\varphi, T)=\emptyset$ (resp.\ $K_0^{\text{L}}(\varphi, T)=\emptyset$) we have that $f^{\csr}(\varphi, T)=\emptyset$ (resp.\ $f^{\lsr}(\varphi, T)=\emptyset$).

\subsection{Consecutive Domains}
Let $\rhd=(a_1, a_2, \dots, a_n)$ be a linear order over~$N$. For each individual $a\in N$, let
\[\varphi_{\rhd}(a)=(\varphi(a, a_1), \varphi(a, a_2), \dots, \varphi(a, a_n)).\]
We say that~$\varphi_{\rhd}(a)$ is {\it{qualifying consecutive}}~(QC) with respect to the order~$\rhd$ if all~$1$s are consecutive in~$\varphi_{\rhd}(a)$, i.e., there are $i, j\in [n]$ such that $i\leq j$, $\varphi(a, a_x)=1$ for all~$x$ such that $i\leq x\leq j$, and $\varphi(a, a_x)=0$ for all other possible values of~$x$.
We say that $\varphi_{\rhd}(a)$ is {\it{disqualifying consecutive}}~(DQC) with respect to ~$\rhd$ if $0$s are consecutive in $\varphi_{\rhd}(a)$.
We say that~$\varphi$ is~QC (resp.\ DQC) if there is at least one linear order~$\rhd$ over~$N$ with respect to which every $\varphi_{\rhd}(a)$ where $a\in N$ is~QC (resp.\ DQC).

It is immediately clear from the above definition that checking if a profile is~QC or~DQC
is equivalent to checking if a $0$-$1$ matrix satisfies the consecutive ones property, which can be done in polynomial-time (cf.~\cite{DBLP:journals/jair/PetersL20,DBLP:journals/jcss/BoothL76,DBLP:journals/jal/Hsu02}).

\subsection{Group Control}
\label{sec-problem-formulations}
Let us now formally state the group control problem we study. Let~$f$ be a social aggregation rule.

\EP
{Group Control by Adding Individuals (GCAI)}
{A $5$-tuple $(N, \varphi, S, T, k)$
    of  a set~$N$ of individuals, a profile~$\varphi$ over~$N$, two nonempty subsets~$S, T \subseteq N$ such
    that~$S \subseteq T$, and an integer~$k$.}
{Is there a subset
    $U\subseteq N \setminus T$ such that~$|U|\leq k$ and
    $S\subseteq f(\varphi, T \cup U)$?
}

In what follows, we call members of~$S$ {\it{distinguished individuals}}.

\subsection{Parameterized Complexity}
A {\it{parameterized problem}} is subset $\Sigma^*\times \mathbb{N}$ where~$\Sigma$ is a fixed alphabet. A parameterized problem is {\fpt} if there is an algorithm so that for each instance $(X, \kappa)$ of the problem the algorithm determines correctly if $(X, \kappa)$ is a {\yesins} in time $f(\kappa)\cdot \abs{X}^{\bigo{1}}$, where~$f$ is a computable function in the parameter~$\kappa$. The following hierarchy has been developed to classify parameterized problems:
\[{\fpt}\subseteq \wa\subseteq \wb\subseteq \cdots \subseteq \xp.\]
A parameterized problem is {\wbh} if all problems in {\wb} are parameterized reducible to the problem. {\wbh} problems do not admit any {\fpt}-algorithms unless the above hierarchy collapses to some level.

A {\it{kernelization}} of an {\fpt} problem~$P$ is an algorithm which takes an instance $(X, \kappa)$ of~$P$ as input and outputs an instance $(X', \kappa')$ of~$P$ such that
\begin{enumerate}
    \item[(1)] the algorithm runs in polynomial time in the size of~$(X, \kappa)$,
    \item[(2)] $(X, \kappa)$ is a {\yesins} if and only if $(X', \kappa')$ is a {\yesins}, and
    \item[(3)] $\abs{X'}\leq g(\kappa)$ for some computable function~$g$ in~$\kappa$.
\end{enumerate}
If~$P$ has a kernelization where~$g$ is a polynomial, we say that~$P$ admits a polynomial kernel.

For further discussions on parameterized complexity, we refer to~\cite{Cygan2015,DBLP:series/txcs/DowneyF13}.

\subsection{Useful Graph Problems}
In the following, we introduce some useful graph problems for our study.
We assume the reader is familiar with the basics in graph theory~\cite{DBLP:books/daglib/0022205,Douglas2000}.

A {\textit{bipartite graph}} is a graph~$G$ whose vertices can be divided into two disjoint sets $R$ and $B$ so that the edges of~$G$ are only between~$R$ and~$B$. A vertex~$v$ {\textit{dominates}} another vertex~$u$ if there is an edge between them. For two disjoint subsets~$A$ and~$B$ of vertices, we say that~$A$ dominates~$B$ if every vertex in~$B$ is dominated by at least one vertex in~$A$.

\EP
{Red-Blue Dominating Set (RBDS)}
{A bipartite graph~$G=(R\cup B,E)$ and an integer~$\kappa$.}
{Is there a subset~$R'\subseteq R$ of at most~$\kappa$ vertices dominating~$B$?}

{\prob{RBDS}} is a well-known {\nph} problem, and from the parameterized complexity point of view it is {\wbh} with respect to~$\kappa$~\cite{DBLP:series/txcs/DowneyF13}. We will use this problem to establish our fixed-parameter intractability results.

Some of our problems are solved by reducing to a variant of the {\prob{Directed Steiner Tree}} problem defined below.
For a graph (resp.\  digraph)~$G$, we use~$V(G)$ to denote its vertex set, and use~$E(G)$ to denote its edge (resp.\ arc) set. For a subset~$J$ of edges (resp.\ arcs) of~$G$,~$V(J)$ is the set of vertices incident with edges (resp.\ arcs)  in~$J$.

\EP
{Directed Steiner Tree (\prob{DST})}
{A digraph~$G$, a subset $X\subseteq V(G)$ of vertices in~$G$ called terminals, a vertex $u\in V(G)\setminus X$  called the root, a function $w: E(G)\rightarrow \mathbb{N}\cup \{0\}$ assigning to each arc an integer weight, and an integer~$p$.}
{Is there a subset $J\subseteq E(G)$  so that

(1) $\sum_{e\in J}w(e)\leq p$; and

(2) for every $x\in X$, there is a directed path from~$u$ to~$x$ in the subgraph of~$G$ induced by~$J$?}

It is known that {\prob{DST} is {\nph} and, moreover, it is {\fpt} with respect to the number of terminals~$\abs{X}$~\cite{DBLP:journals/networks/DreyfusW71,DBLP:journals/siamdm/GuoNS11,DBLP:conf/icde/DingYWQZL07}. More precisely, {\prob{DST}}
 can be solved in $\bigos{2^{\abs{X}}}$ time~\cite{DBLP:journals/networks/DreyfusW71,DBLP:journals/siamdm/GuoNS11,DBLP:conf/stoc/BjorklundHKK07,DBLP:journals/mst/FuchsKMRRW07}\footnote{${\bigos{\cdot}}$ is the ${\bigo{\cdot}}$ notion with the polynomial factor being ignored.}.

\EP
{Directed Vertex Weighted Steiner Tree (DVWST)}
{A digraph~$G$, a subset $X\subseteq V(G)$ of terminals, a vertex $u\in V(G)\setminus X$ called the root, a function $w: V(G)\rightarrow \mathbb{N}\cup \{0\}$ assigning to each vertex an integer weight, and an integer~$p$.}
{Is there a subset $J\subseteq V(G)\setminus (X\cup \{u\})$ so that

(1) $\sum_{v\in J}w(v)\leq p$; and

(2) for every $x\in X$, there is a directed path from~$u$ to~$x$ in the subgraph of~$G$ induced by $J\cup \{u\}\cup X$?}

We shall use {\prob{DVWST}} as an intermediate problem to established our {\fpt}-results. Note that DVWST can be trivially transformed into DAWST in polynomial time by splitting each node into two nodes connected by one arc between them with the same weight as the original node. Therefore, DVWST is also {\fpt} with respect to the number of terminals (Theorem~\ref{thm-dvwst-fpt}). 

\section{Our Results}
\label{sec-results}
We shall first study two {\fpt}-algorithms for {\prob{GCAI}} in the general domain, and then we explore the complexity of {\prob{GCAI}} restricted to the~QC and~DQC domains.

\subsection{The General Domain}
First, we resolve the open questions regarding {\prob{GCAI}} for $f^{\text{CSR}}$ and $f^{\text{LSR}}$ in the affirmative, starting with the one for~$f^{\text{CSR}}$.
To this end, we first show that the {\prob{DVWST}} problem is {\fpt} with respect to the number of terminals.

\begin{theorem}
\label{thm-dvwst-fpt}
{\prob{DVWST}} can be solved in~{$\bigos{2^{\ell}}$} time where~$\ell$ is the number of terminals.
\end{theorem}

The proof of Theorem~\ref{thm-dvwst-fpt} is deferred to the Appendix. We start with a lemma which suggests that to solve {\prob{GCAI}} for~$f^{\text{CSR}}$ we can make a guess on one of the individuals who are
qualified by all individuals in the final profile. This enables us to split an instance of {\prob{GCAI}} for~$f^{\text{CSR}}$ into polynomially many subinstances which are then solved via Theorem~\ref{thm-dvwst-fpt}.

In order to state the lemma,  we need the following notions.
The {\textit{incidence graph}} of a profile~$\varphi$ over~$N$, denoted~$G_{\varphi}$, is the digraph whose vertices are exactly the individuals in~$N$, and there is an arc from $a\in N$ to $a'\in N$ if and only if~$a$  qualifies~$a'$. Note that the incidence graph may contain loops.

\begin{lemma}
\label{lem-csr-property}
Let~$\varphi$ be a profile over a set~$N$ of individuals so that $f^{\memph{CSR}}(\varphi, N)\neq\emptyset$. Let $a\in N$ be an individual qualified by all individuals in~$N$ with respect to~$\varphi$. Then, for every $a'\in N\setminus \{a\}$, it holds that~$a'\in f^{\memph{CSR}}(\varphi, N)$ if and only if there is a directed path from~$a$ to~$a'$ in~$G_{\varphi}$.
\end{lemma}

\begin{proof}
Let $\varphi$,~$a$, and~$a'$ be as stipulated in the lemma.
It is clear from the definition of $f^{\text{CSR}}$ that if there is a directed path from~$a$ to~$a'$, then~$a'\in f^{\text{CSR}}(\varphi, N)$. It remains to show the other direction. Assume that $a'\in f^{\text{CSR}}(\varphi, N)$. Due to the definition of~$f^{\text{CSR}}$, there must be an individual $b\in N$ who is in the initial set of socially qualified individuals, and there is a directed path from~$b$ to~$a'$ in the incidence graph~$G_{\varphi}$ of~$\varphi$. So,~$b$ is qualified by all individuals including~$a$, implying that there is a directed path from~$a$ to~$a'$ in~$G_{\varphi}$.
\end{proof}

Observe that if an individual in~$S$ qualifies another individual in~$S$, then if the former is socially qualified so is the latter. The following reduction rule implements this observation.

\begin{rerule}
\label{reduction-rule}
If there are two distinct individuals $a, a'\in S$  such that~$a$ qualifies~$a'$, move~$a'$ from~$S$ into~$T\setminus S$.
\end{rerule}

For a subset~$Y$ of vertices in a digraph~$G$, let
\[N^-_G(Y)=\{v\in V(G)\setminus Y \setmid \exists(u\in Y)[\arc{v}{u}\in \ARC{G}]\}\] be the set of inneighbors of vertices in~$Y$, and let
\[N^+_G(Y)=\{v\in V(G)\setminus Y \setmid \exists(u\in Y)[\arc{u}{v}\in \ARC{G}]\}\] be the set of outneighbors of vertices in~$Y$. {\textit{Merging}}~$Y$ is the operation that
 creates one vertex~$v_Y$ so that $N^-_{G}(\{v_Y\})=N^-_{G}(Y)$ and $N^+_{G}(\{v_Y\})=N^+_{G}(Y)$, and removes all vertices of~$Y$ from~$G$.
We call~$v_Y$ the {\textit{merging vertex}} of~$Y$.

Armed with Lemma~\ref{lem-csr-property} and Reduction Rule~\ref{reduction-rule}, we are ready to present our first {\fpt}-algorithm.

\begin{theorem}
\label{thm-csr-fpt}
{\prob{GCAI}} for $f^{\memph{CSR}}$ is {\emph\fpt} with respect to the number~$\ell$ of distinguished individuals. More precisely, it can be solved in $\bigos{2^{\ell}}$ time.
\end{theorem}

\begin{proof}
Let $I=(N, \varphi, S, T, k)$ be an instance of {\prob{GCAI}} for~$f^{\text{CSR}}$.
We first exhaustively apply Reduction Rule~\ref{reduction-rule} to~$I$ so that in the resulting instance no individual in~$S$ qualifies another different individual in~$S$.
Then, we split the instance into~$\abs{N}$ subinstances, each of which takes~$I$ and an individual~$a^*\in N$ as input, and determines if there exists $U\subseteq N\setminus T$ of at most~$k$ individuals so that $S\subseteq f^{\text{CSR}}(\varphi, T\cup U)$, $a^*\in T\cup U$, and all individuals in $T\cup U$ qualify~$a^*$. That is,~$a^*$ is our guessed individual who is in the initial set of socially qualified individuals in the final profile. 
Obviously, the original instance~$I$ is a {\yesins} if and only if at least one of the subinstances is a {\yesins}.

Now we focus on solving a subinstance with a guessed individual~$a^*$.
We assume that~$a^*$ is already included in~$T$, since otherwise we simply move~$a^*$ from $N\setminus T$ into~$T$ and decrease~$k$ by one.
As~$a^*$ is supposed to be qualified by all individuals in the final profile, if there is an individual in~$T$ who disqualifies~$a^*$, we directly discard this subinstance and proceed to the next one. Otherwise, we remove from the subinstance all individuals in $N\setminus T$ who disqualify~$a^*$ (this includes deleting them from both~$N$ and from $\varphi$).
We shall solve the subinstance by reducing it to {\prob{DVWST}}.
We create an instance of {\prob{DVWST}} as follows. The digraph~$G$ in the {\prob{DVWST}} instance is obtained from the incidence graph of~$\varphi$ over~$N$ by merging $f^{\text{CSR}}(\varphi, T)$. Obviously,~$a^*\in f^{\text{CSR}}(\varphi, T)$, and so $f^{\text{CSR}}(\varphi, T)$ is nonempty. Let~$u$ denote the merging vertex of $f^{\text{CSR}}(\varphi, T)$, and we let it be the given root of the {\prob{DVWST}} instance.
Furthermore, we let every individual in $N\setminus T$ have weight~$1$, and all the other individuals have weight~$0$. The terminals are those in~$S$, and the weight upper bound is $p=k$.

If there is a subset~$J$ of vertices of total weight at most~$k$ so that there is a directed path from~$u$ to every terminal $a\in S$, then~$a$ is $f^{\text{CSR}}$ socially qualified in $(\varphi, J\cup S)$ due to Lemma~\ref{lem-csr-property}. Moreover, as all individuals in $N\setminus T$ have weight~$1$ and all the other individuals have weight~$0$, we know that~$J$ contains at most~$k$ individuals from $N\setminus T$, implying that $J\cap (N\setminus T)$ is a \yes-certificate for the subinstance. For the opposite direction, if there is a subset $U\subseteq N\setminus T$ of individuals so that for all $a\in S$ it holds that $a\in f^{\text{CSR}}(\varphi, T\cup U)$, then due to Lemma~\ref{lem-csr-property} there is a directed path from~$u$ to~$a$ in the subgraph of~$G$ induced by $T\cup U$. As the total weight of vertices in~$U$ is at most~$k$, the instance of {\prob{DVWST}} is a {\yesins}.

Regarding the running time, let $\ell=\abs{S}$ be the number of distinguished individuals. As there are at most~$\abs{N}$ subinstances to consider and each of them can be solved in~$\bigos{2^{\ell}}$ time (Theorem~\ref{thm-dvwst-fpt}), the whole algorithm runs in~$\bigos{2^{\ell}}$ time.
\end{proof}

An analogous result for the liberal-start-respecting rule exists. In fact, the algorithm in this case is simpler.
We first use a reduction rule to refine the structure of the instance.

\begin{rerule}
\label{reduction-rule-lsr}
If there are $a, a'\in S$  such that~$a$ qualifies~$a'$, move~$a'$ from~$S$ into~$T\setminus S$.
\end{rerule}

Notice that the difference between Reduction Rule~\ref{reduction-rule} and Reduction Rule~\ref{reduction-rule-lsr} is that in the second one~$a$ and~$a'$ may be the same individual. The reason is that under~$f^{\text{LSR}}$ everyone qualifying herself is already socially qualified without needing other individuals' qualifications.

\begin{theorem}
\label{thm-gcai-lsr-fpt}
{\prob{GCAI}} for $f^{\memph{LSR}}$ is {\memph\fpt} with respect to the number~$\ell$ of distinguished individuals. More precisely, it can be solved in~$\bigos{2^{\ell}}$ time.
\end{theorem}

\begin{proof}
Let $I=(N, \varphi, S, T, k)$ be an instance of {\prob{GCAI}} for~$f^{\text{LSR}}$.
We first apply Reduction Rule~\ref{reduction-rule-lsr} to~$I$ iteratively until it does not apply. Then, we solve the instance by reducing it to a {\prob{DVWST}} instance as follows. The digraph~$G$ of the {\prob{DVWST}} instance is obtained from the incidence graph of~$\varphi$ over~$N$ by creating one new vertex~$u$ and creating arcs from~$u$ to everyone in~$N$ qualifying herself (i.e., $\{a\in N \setmid \varphi(a,a)=1\}$).
We set~$u$ as the root, and set~$S$ as the set of the terminals. Finally, we let the weight of all individuals in $N\setminus T$ be~$1$, and those of others be~$0$.
Similar to the analysis in the proof of Theorem~\ref{thm-csr-fpt}, we can show that the two instances are equivalent. The running time of the algorithm follows from Theorem~\ref{thm-dvwst-fpt}.
\end{proof}

Next, we strengthen the {\nphns} of {\prob{GCAI}} for the consensus-start-respecting rule established in~\cite{DBLP:journals/aamas/YangD18} by showing its {\wbhns} with respect to the number of added individuals. We also have a {\wbhns} result for the liberal-start-respecting rule, but we will present it in the next section because this result holds even in a specific domain which is not the focus of this section.

\begin{theorem}
\label{thm-gcai-csr-wbh}
{\prob{GCAI}} for $f^{\memph{CSR}}$ is {\memph\wbh} with respect to the number of added individuals.
\end{theorem}

\begin{proof}
We prove the theorem via a reduction from the {\prob{RBDS}} problem. Let $(G, \kappa)$ be an {\prob{RBDS}} instance, where $G=(R\cup B, E)$ is a bipartite graph with the vertex partition $(R, B)$, and~$\kappa$ is an integer. We create an instance of {\prob{GCAI}} for $f^{\text{CSR}}$ as follows. First, we create for each vertex in~$G$ an individual denoted by the same symbol for notational brevity. Let $N=R\cup B$ and let $S=T=B$. We define a profile~$\varphi$ over~$N$ so that:
\begin{itemize}
    \item each $b\in B$ qualifies all individuals in~$R$ and disqualifies all individuals in~$B$; and
    \item each $r\in R$ qualifies all individuals in~$R$ and, moreover, for each individual $b\in B$ it holds that~$r$ qualifies~$b$ if and only if~$r$ and~$b$ are adjacent in~$G$.
\end{itemize}
The instance of {\prob{GCAI}} for $f^{\text{CSR}}$ is $(N, \varphi, S, T, \kappa)$. The reduction can be done in polynomial time. We show the correctness of the reduction as follows.

$(\Rightarrow)$ Assume that there is a subset $R'\subseteq R$ of at most~$\kappa$ vertices dominating~$B$. According to the definition of~$\varphi$, individuals in~$R$ are qualified by all individuals in~$N$. Therefore, it holds that $R'\subseteq f^{\text{CSR}}(\varphi, B\cup R')$. Moreover, as~$R'$ dominates~$B$, for every~$b\in B$ there is at least one $r\in R$ which dominates~$b$. By the definition of~$\varphi$, $r$ qualifies~$b$, implying that~$b\in f^{\text{CSR}}(\varphi, B\cup R')$. As this holds for all $b\in B$, we know that the instance of {\prob{GCAI}} for $f^{\text{CSR}}$ constructed above is a {\yesins}.

$(\Leftarrow)$ Assume that there is a subset $R'\subseteq R$ (recall that $N\setminus T=R$) of cardinality at most~$\kappa$ so that $B\subseteq f^{\text{CSR}}(\varphi, B\cup R')$. Let~$b$ be any arbitrary individual in~$B$. According to the definition of~$\varphi$,~$b$ is qualified only by individuals in~$R$ who dominate~$b$ in~$G$. As $b\in f^{\text{CSR}}(\varphi, B\cup R')$, this implies that~$R'$ contains at least one vertex dominating~$b$. As this holds for all $b\in B$, we conclude that~$R'$ dominates~$B$. Given $\abs{R'}\leq \kappa$, we conclude that the {\prob{RBDS}} instance is a {\yesins}.
\end{proof}

\subsection{The Consecutive Domains}
Now we explore the complexity of {\prob{GCAI}} for the two procedural rules restricted to the consecutive domains.

\begin{lemma}
\label{lem-property-csr}
Let~$\varphi$ be a profile over~$N$ which is~{\memph{QC}} with respect to a linear order~$\rhd$ of~$N$. Then, all individuals in $f^{\memph{CSR}}(\varphi, N)$ are consecutive in the order~$\rhd$.
\end{lemma}

\begin{proof}
If $f^{\text{CSR}}(\varphi, N)=\emptyset$, the lemma vacuously holds. Otherwise, there is an individual $a\in N$ who is qualified by all individuals. As~$\varphi$ is~QC with respect to~$\rhd$, all individuals only qualify individuals consecutive in~$\rhd$, and they all qualify~$a$. It follows that all individuals in~$f^{\text{CSR}}(\varphi, N)$ are consecutive in~$\rhd$.
\end{proof}

Based on Lemma~\ref{lem-property-csr}, we can derive a polynomial-time algorithm for {\prob{GCAI}} for~$f^{\text{CSR}}$.

\begin{theorem}
\label{poly-gcai-csr-qc}
{\prob{GCAI}} for $f^{\memph{CSR}}$ is polynomial-time solvable when restricted to~{\memph{QC}} profiles.
\end{theorem}

\begin{proof}
Let $I=\left( N,\varphi ,S,T,k\right) $ be an instance of {\prob{GCAI}} for~$f^{\text{CSR}}$ where~$\varphi$ is~QC with respect to a linear order~$\rhd$ over~$N$.
Let~$a_i$ and~$a_j$ be respectively the left-most and the right-most individuals in~$\rhd$ that are from~$S$.
Due to Lemma~\ref{lem-property-csr}, the question of~$I$ is equivalent to making~$a_i$ and~$a_j$ socially qualified in~$T$ by adding at most~$k$ individuals from $N\setminus T$ into~$T$. By light of this fact, we move all individuals except~$a_i$ and~$a_j$ from~$S$ into~$T\setminus S$. After this operation,~$S$ contains at most two individuals ($S$ is a singleton when $i=j$). Then, we solve the instance in polynomial time by Theorem~\ref{thm-csr-fpt}.
\end{proof}

Now we move on to the liberal-start-respecting rule. Unlike~$f^{\text{CSR}}$, we show that {\prob{GCAI}} for~$f^{\text{LSR}}$ remains computationally hard even when restricted to~QC profiles.

\begin{theorem}
\label{thm-gcai-lsr}
{\prob{GCAI}} for $f^{\memph{LSR}}$ is {\memph\nph} and is {\memph\wbh} with respect to the number of added individuals even when restricted to~{\memph{QC}} profiles.
\end{theorem}

\begin{proof}
We prove the theorem by giving a reduction from the {\prob{RBDS}} problem to {\prob{GCAI}} for~$f^{\text{LSR}}$ restricted to~QC profiles.
Let $(G, \kappa)$ be an instance of {\prob{RBDS}}, where $G=(B\cup R, E)$ is a bipartite graph, and~$\kappa$ is an integer.
For each $b\in B$, we construct an individual denoted still by~$b$ for notational simplicity. For each $r\in R$, let~$d(r)$ be the degree of~$r$ in~$G$. For each $r\in R$, we construct $d(r)+1$ individuals~$r(0),~r(1), \dots, r(d(r))$. Let $C(r)=\{r(i) \setmid i\in [d(r)]\}$ for each $r\in R$, and let $C(R)=\bigcup_{r\in R}C(r)$. In addition, let~$N$ denote the set of the above constructed $\abs{B}+\abs{R}+\sum_{r\in R}d(r)$ individuals, let $S=B$, and let $T=B\cup C(R)$. We define a profile~$\varphi$ over~$N$ as follows.
\begin{itemize}
    \item For each red vertex $r\in R$, the individual~$r(0)$ qualifies~$r(0)$,~$r(1)$, $\dots$, $r(d(r))$, and each~$r(1)$, $r(2)$, $\dots$, $r(d(r))$ qualifies exactly one neighbor of~$r$ in~$G$ so that every neighbor of~$r$ is qualified by exactly one of these~$d(r)$ individuals.
    \item For each $a, a'\in N$ where $\varphi(a, a')$ is not specified above, we define $\varphi(a, a')=0$. 
\end{itemize}
The instance of {\prob{GCAI}} is $( N,\varphi, S, T, \kappa)$.

It is easy to see that the profile~$\varphi$ is~QC. In fact, except those in $\{r(0) \mid r\in R\}$, all the other individuals qualify at most one individual. Moreover, as every~$r(0)$ where $r\in R$ qualifies exactly the $d(r)+1$ individuals created for~$r$, the profile is~QC with respect to any linear order over~$N$ where for every $r\in R$ the $d+1$ individuals created for~$r$ are consecutive.

The construction takes polynomial time. In the following, we show the correctness of the reduction.

$(\Rightarrow)$ Assume that there is a subset $R'\subseteq R$ of at most~$\kappa$ vertices dominating~$B$. Let $U=\{r(0) \setmid r\in R'\}$. We show that $S\subseteq f^{\text{LSR}}(\varphi, T\cup U)$. Note that as $\varphi(r(0), r(0))=1$ for all $r\in R$, and~$r(0)$ qualifies also all the other individuals created for~$r$, we know that for every $r\in R'$, the individuals~$r(0)$,~$r(1)$,~$\dots$,~$r(d(r))$ are all $f^{\text{LSR}}$ socially qualified in $T\cup U$ at~$\varphi$. Let~$b$ be an individual in~$S$. As~$R'$ dominates~$B$ and $B=S$,~$b$ has at least one neighbor $r\in R'$ in~$G$. Then, due to the above construction, there exists an individual~$r(i)$ where $i\in [d(r)]$ who qualifies~$b$. As~$r(i)\in f^{\text{LSR}}(\varphi, T\cup U)$, it follows that~$b\in f^{\text{LSR}}(\varphi, T\cup U)$. As this holds for all $b\in S$, the above constructed instance of~{\prob{GCAI}} for~$f^{\text{LSR}}$ is a {\yesins}.

$(\Leftarrow)$ Assume that there is a $U\subseteq N\setminus T$ such that $\abs{U}\leq \kappa$ and~$S\subseteq f^{\text{LSR}}(\varphi, T\cup U)$. Let $R'=\{r\in R \setmid r(0)\in U\}$. Clearly,~$\abs{R'}= \abs{U}\leq \kappa$. We claim that~$R'$ dominates~$B$. Let~$b$ be a vertex (individual) in~$B$. By the definition of~$\varphi$,~$b$ is only qualified by individuals in $C(r)$ such that $r\in R$ and~$r$ dominates~$b$. Then, as~$b\in f^{\text{LSR}}(\varphi, T\cup U)$, there exist $r\in R$ and $i\in [d(r)]$ such that $r(i)\in f^{\text{LSR}}(\varphi, T\cup U)$ and $\varphi(r(i),b)=1$. Note that the only individual who qualifies~$r(i)$ is the individual~$r(0)$ who qualifies herself. This means that $r(0)\in U$, and hence $r\in R'$, further implying that~$b$ is dominated by~$r$. As the above argument holds for all $b\in B$, we conclude that~$R'$ dominates~$B$.
\end{proof}

When restricted to~DQC, we can show that {\prob{GCAI}} for both procedural rules are polynomial-time solvable. A crucial observation is that if the given instance is a {\yesins}, we need at most two individuals to bring all distinguished individuals into the set of socially qualified individuals.

\begin{theorem}
\label{thm-gcai-lsr-dqc}
{\prob{GCAI}} for $f^{\memph{CSR}}$  and {\prob{GCAI}} for $f^{\memph{LSR}}$ restricted to~{\memph{DQC}} profiles are polynomial-time solvable.
\end{theorem}

\begin{proof}
Let $I=(N, \varphi, S, T, k)$ be an instance of {\prob{GCAI}} for~$f^{\text{CSR}}$ (resp.~$f^{\text{LSR}}$), where~$\varphi$ is~DQC with respect to a linear order $\rhd=(a_{1},a_{2},\ldots ,a_{n})$ over~$N$. If~$k< 2$, we solve the instance in polynomial time by a brute-force search. So, in the following, let us assume that $k\geq 2$.
By $a\unrhd a'$, we mean either $a\rhd a'$ or $a=a'$. Let~$N'=\{a\in N \setmid \exists(a'\in N)[\varphi(a, a')=0]\}$ be the set of individuals in~$N$ disqualifying at least one individual in~$N$. For each $a\in N'$, let~$L(a)$ be the left-most individual~$a$ disqualifies, and let~$R(a)$ be the right-most individual~$a$ disqualifies in~$\rhd$. More precisely, $L(a)=a_j$ (resp.\ $R(a)=a_j$) such that $\varphi(a, a_j)=0$ and, moreover, for all $a_i\in N$ such that $\varphi(a, a_i)=0$ it holds that $i\geq j$ (resp. $i\leq j$). Let
\[A=\{a\in N' \mid \varphi(a, a_n)=1, \forall{(a'\in N, \varphi(a', a_n)=1)}[R(a) \unrhd R(a')]\}\]
and let
\[B=\{a\in N' \mid \varphi(a, a_0)=1, \forall{(a'\in N, \varphi(a', a_1)=1)}[L(a') \unrhd L(a)]\}.\]

For every subset $Z\subseteq N$, let
\[{\bf{1}}_{\varphi}(Z)=\{a\in N\setmid \exists(a'\in Z)[\varphi(a', a)=1]\}\]
denote the set of individuals qualified by at least one individual from~$Z$.

\begin{observation}
\label{obs-1}
Let~$X$ be a subset of $A\cup B$ of cardinality at most two so that for each $Y\in \{A, B\}$ it holds that $\abs{X\cap Y}=1$ whenever $Y\neq\emptyset$. Then, if $N\setminus N'=\emptyset$, for every~$Z\subseteq N$, it holds that~${\bf{1}}_{\varphi}(Z)\subseteq {\bf{1}}_{\varphi}(X)$. Moreover, if $N\setminus N'\neq \emptyset$, for every~$Z\subseteq N$, it holds that~${\bf{1}}_{\varphi}(Z)\subseteq {\bf{1}}_{\varphi}(\{a\})$ for every~$a\in N\setminus N'$.
\end{observation}

In view of Observation~\ref{obs-1}, if~$I$ is a {\yesins}, there exists a subset of~$N$ of at most two individuals so that every individual in~$S$ is qualified by at least one individual in the subset. Therefore, to solve~$I$, we enumerate all subsets~$S'\subseteq N$ of at most two individuals so that~$S\subseteq {\bf{1}}_{\varphi}(S')$. For each enumerated subset~$S'$, we solve an instance $I_{S'}=(N, \varphi, S', T', k')$ of {\prob{GCAI}}, where~$T'=T\cup S'$ and $k'=k-\abs{S'\cap (N\setminus T)}$. By Theorems~\ref{thm-csr-fpt} and~\ref{thm-gcai-lsr-fpt}, each~$I_{S'}$ can be done in polynomial time. The original instance~$I$ is a {\yesins} if and only if there exists at least one enumerated~$S'$ so that~$I_{S'}$ is a {\yesins}.
As we have at most $\abs{N}^2$ enumerations, the whole algorithm takes polynomial time.
\end{proof}

\section{Concluding Remarks}
\label{sec_conclusion}

We proved that {\prob{GCAI}} for both the consensus-start-respecting rule ($f^{\text{CSR}}$) and the liberal-start-respecting rule ($f^{\text{LSR}}$) are {\fpt} with respect to the number of distinguished individuals (Theorems~\ref{thm-csr-fpt} and~\ref{thm-gcai-lsr-fpt}), resolving two open questions left in~\cite{DBLP:journals/aamas/YangD18}. Additionally, we showed that {\prob{GCAI}} for~$f^{\text{CSR}}$ and {\prob{GCAI}} for~$f^{\text{LSR}}$ are {\wbh} with respect to the solution size~$k$ (Theorems~\ref{thm-gcai-csr-wbh} and~\ref{thm-gcai-lsr}). Furthermore, we studied {\prob{GCAI}} restricted to the qualifying consecutive~(QC) domain and the disqualifying consecutive~(DQC) domain. We showed that both {\prob{GCAI}} for $f^{\text{CSR}}$ and {\prob{GCAI}} for $f^{\text{LSR}}$ become polynomial-time solvable when restricted to the~DQC domain  (Theorem~\ref{thm-gcai-lsr-dqc}). However, when restricted to the~QC domain, {\prob{GCAI}} for~$f^{\text{CSR}}$ is polynomial-time solvable (Theorem~\ref{poly-gcai-csr-qc}), while {\prob{GCAI}} for~$f^{\text{LSR}}$ turned out to be computationally hard (Theorem~\ref{thm-gcai-lsr}). See Table~\ref{tab-summary} for a summary of these results.

\begin{table}[ht]
    \caption{A summary of our main results regarding the parameterized complexity of group control by adding individuals for the two procedural rules $f^{\text{CSR}}$ and $f^{\text{LSR}}$. Here,~$\ell$ denotes the number of distinguished candidates, and~$k$ denotes the solution size.}\medskip
    \label{tab-summary}
       \centering{
    \begin{tabular}{|l|c|c|c|c|}\hline
         & \multicolumn{2}{c|}{parameters} & \multicolumn{2}{c|}{restricted domains} \\ \cline{2-5}
         & $\ell$ & $k$ & QC & DQC \\ \hline
    $f^{\text{CSR}}$  & {\fpt} (Thm.~\ref{thm-csr-fpt}) & {\wbh} (Thm.~\ref{thm-gcai-csr-wbh}) & {\poly} (Thm.~\ref{poly-gcai-csr-qc}) & {\poly} (Thm.~\ref{thm-gcai-lsr-dqc}) \\ \hline
    $f^{\text{LSR}}$ & {\fpt} (Thm.~\ref{thm-gcai-lsr-fpt}) & {\wbh} (Thm.~\ref{thm-gcai-lsr}) & {\wbh} (Thm.~\ref{thm-gcai-lsr})  & {\poly} (Thm.~\ref{thm-gcai-lsr-dqc}) \\ \hline
    \end{tabular}
    }
\end{table}

Given the fixed-parameter tractability of {\prob{GCAI}} stated in Theorems~\ref{thm-csr-fpt} and~\ref{thm-gcai-lsr-fpt}, one may wonder whether the two problems admit polynomial kernels when parameterized by the number of distinguished individuals.
Regarding this issue, we remark that both reductions in the proofs of Theorems~\ref{thm-gcai-csr-wbh} and \ref{thm-gcai-lsr} are in fact polynomial parameter transformations with respect to the combined parameter $\abs{S}+k$ of the number of distinguished candidates and the number of added individuals. Then, by the lower bound technique developed by Dom, Lokshtanov, and Saurabh~\cite{DBLP:journals/talg/DomLS14}, we have the following two corollaries refuting the possibility of the existence of polynomial kernels for the two problems.

\begin{corollary}
\label{cor-1}
{\prob{GCAI}} for $f^{\memph{CSR}}$ does not admit any polynomial kernel with respect to the parameter $\abs{T}+k$ unless the polynomial hierarchy collapses to the third level $\left({\textsf{PH}}=\Sigma_{{\textsf{P}}}^{\emph{\text{3}}}\right)$.
\end{corollary}

Note that as $S\subseteq T$, Corollary~\ref{cor-1} also means that {\prob{GCAI}} for $f^{\memph{CSR}}$ is unlikely to admit any polynomial kernel with respect to~$\abs{S}+k$.

\begin{corollary}
{\prob{GCAI}} for $f^{\memph{LSR}}$ does not admit any polynomial kernel with respect to the parameter $\abs{S}+k$ unless the polynomial hierarchy collapses to the third level. Moreover, this holds even when restricted to~{\memph{QC}} profiles.
\end{corollary}

Additionally, note that {\prob{GCAI}} is {\fpt} with respect to $t=\abs{N\setminus T}$ because it can be solved in $\bigos{2^{t}}$ time by a brute-force search. Because {\prob{RBDS}} is unlikely to admit any polynomial kernel with respect to~$\abs{R}$~\cite{Cygan2015}, our reductions in Theorems~\ref{thm-gcai-csr-wbh} and \ref{thm-gcai-lsr} respectively lead to the following two corollaries.

\begin{corollary}
{\prob{GCAI}} for $f^{\memph{CSR}}$ does not admit any polynomial kernel with respect to the parameter $\abs{N\setminus T}$ unless the polynomial hierarchy collapses to the third level.
\end{corollary}

\begin{corollary}
{\prob{GCAI}} for $f^{\memph{LSR}}$ does not admit any polynomial kernel with respect to the parameter $\abs{N\setminus T}$ unless the polynomial hierarchy collapses to the third level. Moreover, this holds even when restricted to QC profiles.
\end{corollary}

Finally, observe that {\prob{GCAI}} can be also solved in $\bigos{t^k}$ time by a brute-force search. As {\prob{RBDS}} cannot be solved in $\bigos{2^{\smallo{\abs{R}}}}$ time assuming the Strong Exponential Time Hypothesis (SETH)~\cite{Cygan2015}, and it cannot be solved in $\bigos{\abs{R}^{\kappa}}$ time either assuming ETH~\cite{DBLP:journals/jcss/ChenHKX06}, our reductions in the proofs of Theorems~\ref{thm-gcai-csr-wbh} and~\ref{thm-gcai-lsr} imply that these brute-force based algorithms are essentially optimal.

\begin{corollary}
Unless SETH fails {\prob{GCAI}} for $f^{\memph{CSR}}$ cannot be solved in $\bigos{2^{\smallo{t}}}$ time, and unless ETH fails {\prob{GCAI}} for $f^{\memph{CSR}}$ cannot be solved in~$\bigos{t^{\smallo{k}}}$ time.
\end{corollary}

\begin{corollary}
\label{cor-last}
Unless SETH fails {\prob{GCAI}} for $f^{\memph{LSR}}$ cannot be solved in $\bigos{2^{\smallo{t}}}$ time, and unless ETH fails {\prob{GCAI}} for $f^{\memph{LSR}}$ cannot be solved in~$\bigos{t^{\smallo{k}}}$ time. Moreover, this holds even when restricted to QC profiles.
\end{corollary}

\section*{Appendix}

\begin{proof}[Proof of Theorem~\ref{thm-dvwst-fpt}]
Let  $I=(G, X, u, w, p)$ be an instance of {\prob{DVWST}}. We create an instance of {\prob{DST}} equivalent to~$I$ as follows.

We first create an arc-weighted digraph~$G'$ obtained from~$G$ by performing the following operations:
\begin{enumerate}
    \item[(1)]  Replace every vertex $v\in V(G)\setminus (X\cup \{u\})$ with two vertices~$v^{\text{in}}$ and~$v^{\text{out}}$, add an arc from $v^{\text{in}}$ to~$v^{\text{out}}$ with weight~$w(v)$, and add some arcs so that the inneighbors of $v^{\text{in}}$ are exactly the inneighbors of $v$ in~$G$, and the outneighbors of $v^{\text{out}}$ are exactly the outneighbors of~$v$ in~$G$.
   \item[(2)] Set the weight of all the arcs whose weights are not yet specified to be~$0$.
\end{enumerate}
Let $w': E(G')\rightarrow \mathbb{N}\cup \{0\}$ be the function corresponding to the weights specified for the arcs in~$G'$ above. The instance of {\prob{DST}} is $(G', X, u, w', p)$. The reduction clearly can be done in polynomial time. It remains to show the correctness.

$(\Rightarrow)$ Assume that the {\prob{DVWST}} instance is a {\yesins}, i.e., there is a subset $J\subseteq V(G)\setminus (X\cup\{u\})$ such that $\sum_{v\in J}w(J)\leq p$, and for every terminal $x\in X$ there is a directed path from~$u$ to~$x$ in the subgraph of~$G$ induced by~$J\cup X\cup \{u\}$. Let~$J'=E(G)\cup \{\arc{v^{\text{in}}}{v^{\text{out}}} \setmid v\in J\}$. Due to the above construction every original arc in~$G$ has weight~$0$ under~$w'$, and every arc $\arc{v^{\text{in}}}{v^{\text{out}}}$ has weight~$w(v)$ under~$w'$. It follows that
\[\sum_{ \arc{v^{\text{in}}}{v^{\text{out}}} \in J'} w'(\arc{v^{\text{in}}}{v^{\text{out}}})=\sum_{v\in J}w(v)\leq p.\]
Moreover, if $u\ v_1\ v_2\ \cdots\ v_t\ x$ is a directed path from the root~$u$ to some terminal $x\in X$ in the digraph~$G$, by the definition of~$G'$ we know that $u\ v_1^{\text{in}}\ v_1^{\text{out}}\ v_2^{\text{in}}\ v_2^{\text{out}}\ \cdots\  v_t^{\text{in}}\ v_t^{\text{out}}\ x$ is a directed path in~$G'$. Therefore, the constructed {\prob{DST}} instance is a {\yesins}.

$(\Leftarrow)$ Assume that the constructed instance of {\prob{DST}} is a {\yesins}, i.e., there is a subset~$J$ of arcs in~$G'$ so that $\sum_{e\in J}w'(e)\leq p$ and, moreover, for every terminal $x\in X$ there is a directed path from~$u$ to~$x$ in the subgraph of~$G'$ induced by~$V(J)$. Let
\[J'=\{v\in V(G)\setminus (X\cup \{u\}) \setmid \arc{v^{\text{in}}}{v^{\text{out}}}\in J\}.\]
Similar to the above analysis, we know that $\sum_{v\in J'}w(v)\leq p$ and for every terminal $x\in X$ we can change any directed path from~$u$ to~$x$ in the subgraph of~$G'$ induced by~$V(J)$ into a directed path from~$u$ to~$x$ in the subgraph of~$G$ induced by~$J'\cup X\cup \{u\}$. In particular, observe that each~$v^{\text{in}}$ has a unique outneighbor~$v^{\text{out}}$. So, in any~$u$-$x$ directed path containing a vertex $v^{\text{in}}$, the next vertex after~$v^{\text{in}}$ must be $v^{\text{out}}$. Therefore, from a directed path from~$u$ to~$x$ in~$G'$, we can obtain a directed path from~$u$ to~$x$ in~$G$ by replacing every arc~$\arc{v^{\text{in}}}{v^{\text{out}}}$ in the path with the vertex~$v$.

The theorem follows from the above reduction and the fact that {\prob{DST}} can be solved in~$\bigos{2^{\ell}}$ time, where~$\ell$ is the number of terminals~\cite{DBLP:journals/networks/DreyfusW71,DBLP:journals/siamdm/GuoNS11,DBLP:conf/stoc/BjorklundHKK07,DBLP:journals/mst/FuchsKMRRW07}.
\end{proof}

\end{document}